\documentclass[11pt, reqno]{amsart}       
\usepackage{graphicx}

\usepackage{amsmath}

\usepackage{amssymb}
\usepackage{pifont}
\newtheorem{thm}{Theorem}
\newtheorem{lem}{Lemma}

\newtheorem{cor}{Corollary}
\theoremstyle{definition}

\newtheorem*{thm1}{Theorem}

\newtheorem*{rem2}{Remark}

\newtheorem*{acknowledgement}{Acknowledgement}

\newcommand{\vertiii}[1]{{\left\vert\kern-0.25ex\left\vert\kern-0.25ex\left\vert
		#1 \right\vert\kern-0.25ex\right\vert\kern-0.25ex\right\vert}}

\def \lim   {\text {\rm lim}}

\begin{document}
	
	\title[]{Universal recoverability of quantum states in tracial von-Neumann algebras}
	
	\author{Saptak Bhattacharya}
	
	\address{Indian Statistical Institute\\
		New Delhi 110016\\
		India}
	\email{saptak21r@isid.ac.in}
	
	
	
	
	\begin{abstract}
		In this paper, we discuss a refinement of quantum data processing inequality for the sandwiched quasi-relative entropy $\mathcal{S}_2$ on a tracial von-Neumann algebra. The main result is a universal recoverability bound with the Petz recovery map, which was previously obtained in the finite dimensional setup. 
	\end{abstract}
	
	\date{}
	\maketitle
	
	\section{introduction} Quantum entropies are used as discrimination measures between quantum states. They are widely studied both in the finite and infinite dimensional contexts. An extremely important example is the Kullback-Liebler divergence, defined for density matrices $A$ and $B$ with $\mathrm{supp}\, A\subset\mathrm{supp}\, B$ by $$D(A|B)=\mathrm{tr}\, A(\ln\, A-\ln\, B).$$  Lindblad in \cite{gl} proved that for any completely positive, trace preserving map (also called a quantum channel) $\phi:M_n(\mathbb{C})\to M_k(\mathbb{C})$, \[D(\phi(A)|\phi(B))\leq D(A|B).\label{e1}\tag{1}\] This is known as the {\it data processing inequality} or DPI. Over time, several other entropies and their data processing inequalities have been introduced and studied, namely, the R\'enyi divergences \cite{gl2, pt}, the $f-$ divergences \cite{pt} and the sandwiched quasi-relative entropies \cite{ml, ww, sb}.
	\medskip
	
	It is natural to ask when equality holds in $\eqref{e1}$. The following was proved by Petz (see \cite{hp, pr}) :
	\medskip
	
	\begin{thm1}Let $B\in M_n(\mathbb{C})$ be a density matrix and let $\phi:M_n(\mathbb{C})\to M_k(\mathbb{C})$ be a quantum channel. Let $\mathcal{R}:M_k(\mathbb{C})\to M_n(\mathbb{C})$ be given by $$\mathcal{R}(Y)=B^{1/2}\phi^*(\phi(B)^{-1/2}Y\phi(B)^{-1/2})B^{1/2}.$$ Then, for some density matrix $A\in M_n(\mathbb{C})$, $D(A|B)=D(\phi(A)|\phi(B))$ if and only if $\mathcal{R}(\phi(A))=A$. \end{thm1}
	\medskip
	
	\begin{rem2}The map $\mathcal{R}$ above automatically satisfies $\mathcal{R}(\phi(B))=B$. It is famously called the {\it Petz recovery map}.\end{rem2}
	\medskip
	
	Over time, versions of Petz's theorem have been proved for some specific $f$-divergences \cite{hp}, and for all sandwiched quasi-relative entropies \cite{jen}. These theorems tell us that the {\it universal} recovery map $\mathcal{R}$ recovers any state $A$ for which the distinguishability information within the entropy is preserved. However, perfect recovery cannot be expected in real life, which is why it became necessary to prove refinements of the DPI for these entropies capturing stability.
	\medskip
	
	This turned out to be a difficult problem. Seshadreesan, Wilde and Berta conjectured in 2015 that  \[-2\ln \big[F(A|\mathcal{R}(\phi(A)))\big]\leq D(A|B)-D(\phi(A)|\phi(B)).\label{e2}\tag{2}\] where $D$ is the $K-L$ divergence and $F$ is the fidelity given by $$F(A|B)=\mathrm{tr}\,|A^{1/2}B^{1/2}|$$ for density matrices $A$ and $B$. Later, Junge et. al. \cite{jw} proved the existence of a recovery map $\mathcal{R}^{\prime}$ obtained as an average of {\it rotated} Petz maps such that  \[-2\ln \big[F(A|\mathcal{R}^{\prime}(\phi(A)))\big]\leq D(A|B)-D(\phi(A)|\phi(B)).\] This was later generalized to infinite dimensions by Faulkner et. al. \cite{fh, fh2}.
	\medskip
	
	However, these results lacked a bound with the Petz map itself. The conjectured inequality $\eqref{e2}$ remained unsolved until recently, when it was disproved by us \cite{sap}. 
	\medskip
	
	In 2020, Carlen and Vershynina \cite{cv} proved the first approximate recoverability bound with the Petz map and the $K-L$ divergence. Subsequent work (see \cite{ver, gw}) expanded their novel techniques and generalized it to $f$-divergences and optimized $f$-divergences. However, these results was not universal since they involved an unbounded factor dependent on $A$ before the entropy difference $D(A|B)-D(\phi(A)|\phi(B))$
	\medskip
	
	Later, Cree and Sorce worked with the sandwiched quasi-relative entropy $\mathcal{S}_2$ and proved the bound \[\begin{aligned}4\big[1-F(A|\mathcal{R}\circ\phi(A))\big]^2&\leq ||A-\mathcal{R}\circ\phi(A)||_1^2\\&\leq ||B||^2_2||B^{-1}||[\mathcal{S}_2(A|B)-\mathcal{S}_2(\phi(A)|\phi(B))].\end{aligned}\] This was the first universal stability result with the Petz map. Recently, Gao et. al. in \cite{lg} and we in \cite{sap} independently refined the inequality and obtained the much simpler looking bound : \[4[1-F(A|\mathcal{R}\circ\phi(A))]^2\leq||A-\mathcal{R}\circ\phi(A)||_1\leq [\mathcal{S}_2(A|B)-\mathcal{S}_2(\phi(A)|\phi(B))].\label{e3}\tag{3}\] 
	\medskip
	
	Our proof of inequality $\eqref{e3}$ in \cite{sap} is based on the Araki-Masuda Hilbert space geometry underlying $\mathcal{S}_2$. As we demonstrate in this paper, the advantage of this technique is that we can use it to generalize inequality $\eqref{e3}$ to infinite dimensions. To the best of our knowledge, this is the first {\it universal} approximate recoverability bound in infinite dimensions with the Petz map.
	\medskip
	
	We work on a von-Neumann algebra $\mathcal{M}$ over a Hilbert space $\mathcal{H}$ with a normal, faithful tracial state $\tau:\mathcal{M}\to\mathbb{C}$ (see \cite{kad2} for details). $\tau$ induces a family of $p$-norms for $p\geq 1$ given by $$||X||_p=\tau(|X|^p)^{1/p}$$ for all $X\in\mathcal{M}$. The completion of $\mathcal{M}$ with respect to $||.||_p$ is the corresponding non-commutative $L^p$ space. A concrete description of this space can be found in \cite{nel, sil}, we give a brief summary here. 
	\medskip
	
	Consider a closed densely defined operator $A:D(A)\to\mathcal{H}$. Take the polar decomposition $A=V|A|$. $A$ is said to be {\it affiliated} to $\mathcal{M}$ if 
	the partial isometry $V$ and all the spectral projections of $|A|$ lie in $\mathcal{M}$. Given a Borel $E\subset\mathbb{R}$ let $P(E)$ be the corresponding spectral projection of $|A|$. Write $$|A|=\int_0^{\infty}\lambda dP(\lambda).$$ Then the non-commutative $L^p$ space $L^p(\mathcal{M}, \tau)$ consists of all closed operators $A$ affiliated to $\mathcal{M}$ such that $$\int_{0}^{\infty}\lambda^p dP(\lambda)\,\textless\,\infty.$$ The norm extends naturally to $L^p(\mathcal{M},\tau)$ by $$||A||_p=[\int_0^{\infty}\lambda^p dP(\lambda)]^{1/p}.$$ $\tau$ itself  extends similarly. This makes $L^p(\mathcal{M},\tau)$ a Banach space, and a Hilbert space for $p=2$.  
	\medskip
	
	Let $\psi:M\to\mathbb{C}$ be a normal state. Then there exists a positive $A\in L^1(\mathcal{M},\tau)$ such that $\tau(A)=1$ and $\psi(X)=\tau(AX)$ for all $X\in\mathcal{M}$. These representatives play the role of density matrices. Let $B\in\mathcal{M}$ be positive and invertible with $\tau(B)=1$. Given any positive $A\in L^p(\mathcal{M},\tau)$ with $\tau(A)=1$ we define the sandwiched quasi-relative entropy $\mathcal{S}_p$ by $$\mathcal{S}_p(A|B)=\tau[(B^{-1/2q}AB^{-1/2q})^p]$$ where $q$ is the harmonic conjugate of $p$ satisfying $\frac{1}{p}+\frac{1}{q}=1$.
	\medskip
	
	Like the case for matrices, we model quantum channels by completely positive, trace preserving maps $\phi$ from $\mathcal{M}$ to some other von-Neumann algebra $\mathcal{N}$ with a faithful, normal tracial state $\tau^{\prime}$. The DPI for $\mathcal{S}_p$ has been discussed in \cite{jencova}.
	\medskip
	
	Given a quantum channel $\phi:(M\tau)\to(N,\tau^{\prime})$ we can extend it to the Hilbert space $L^2(\mathcal{M},\tau)$. Its adjoint $\phi^*$ maps $L^2(\mathcal{N},\tau^{\prime})$ to $L^2(M,\tau)$. $\phi$ is said to be strictly CPTP if it takes positive invertible elements to postive invertible elements. Given a strict CPTP map $\phi:(\mathcal{M},\tau)\to(\mathcal{N},\tau^{\prime})$ and a positive, invertible $B\in \mathcal{M}$, define the Petz recovery map $\mathcal{R}:L^2(\mathcal{N},\tau^{\prime})\to L^2(\mathcal{M},\tau)$ by \[\mathcal{R}(Y)=B^{1/2}\phi^*(\phi(B)^{-1/2}Y\phi(B)^{-1/2})B^{1/2}.\label{e4}\tag{4}\]

	Note that even though $\phi^*(\phi(B)^{-1/2}Y\phi(B)^{-1/2})$ might be unbounded, multiplication on both sides with a {\it bounded} operator makes sense in the non-commutative $L^2$ limit. 
	\medskip
	
	In this paper, which we keep as self-contained as possible, the inequality \[||A-\mathcal{R}(\phi(A))||^2_1\leq \mathcal{S}_2(A|B)-\mathcal{S}_2(\phi(A)|\phi(B))\label{e5}\tag{5}\] is proved for any positive $A\in L^2(\mathcal{M},\tau)$ with $\tau(A)=1$. Note that the left-hand side makes sense since $A-\mathcal{R}(\phi(A))\,\in\, L^2(M,\tau)\subset L^1(\mathcal{M},\tau)$. This generalizes the right hand side of inequality $\eqref{e3}$ to infinite dimensions, giving a universal recoverability bound with the Petz map. 
	\medskip
	
One might now wonder whether the left hand side generalizes as well. For that, we need to use an infinite dimensional version of the fidelity. Given two states $\psi_1$ and $\psi_2$ on a unital $C^*$-algebra $\mathcal{A}$, the {\it transition probability} was defined by Uhlmann \cite{uhl2} as \[
\begin{aligned}P(\psi_1,\psi_2)
&:= \sup_{\pi\in\mathrm{Hom}(\mathcal{A}, B(\mathcal{H}))}\{|\langle x, y\rangle|^2:\psi_1(A)=\langle\pi(A)x,x\rangle,\,\\&\psi_2(A)=\langle\pi(A)y, y\rangle,\, A\in\mathcal{A},\,||x||=||y||=1 \}.
\end{aligned}\label{e6}\tag{6}\]

The supremum runs over all possible {\it common} GNS representations of the states $\psi_1$ and $\psi_2$. The fidelity $F(\psi_1|\psi_2)$ is then defined as $\sqrt{P(\psi_1|\psi_2)}$ (see \cite{uhal}). This reduces to the fidelity between density matrices in the finite dimensional setup. We discuss $F(\psi_1|\psi_2)$ in the appendix, where we give an elementary proof of joint concavity in this setup.
\medskip

The Uhlmann fidelity for states helps us deduce the inequality \[4[1-F(\psi_{A}|\psi_{\,\mathcal{R}(\phi(A)}))]^2\leq||A-\mathcal{R}\circ\phi(A)||^2_1\leq [\mathcal{S}_2(A|B)-\mathcal{S}_2(\phi(A)|\phi(B))]\] where $\psi_A$ and $\psi_{\,\mathcal{R}(\phi(A))}$ are the normal states corresponding to $A$ and $\mathcal{R}(\phi(A))$ respectively. This is the main result of this paper, giving a universal approximate recoverability bound in infinite dimensions. 

	\section{Main results}
	
	We need to be careful with the technical details in infinite dimensions. The first step is to show that a quantum channel $\phi$ is indeed $L^2$-continuous, so that the adjoint makes sense.
	\medskip
	
	\begin{thm}\label{t1}Let $(\mathcal{M},\tau)$ and $(\mathcal{N},\tau^{\prime})$ be tracial von-Neumann algebras and let $\phi:(\mathcal{M},\tau)\to(\mathcal{N},\tau^{\prime})$ be CPTP. Then $\phi$ is $L^2$-bounded.\end{thm}
	\begin{proof}Let $X\in\mathcal{M}$. The block matrix \[\begin{pmatrix}X^*X & X^*\\X & 1\end{pmatrix}\] is positive and therefore, \[\begin{pmatrix}\phi(X^*X) & \phi(X)^*\\\phi(X) & \phi(1)\end{pmatrix}\] is positive. Choose $K\in\mathcal{N}$ such that $||K||\leq 1$ and $$\phi(X)=\phi(X^*X)^{1/2}K\phi(1)^{1/2}.$$ Then \[\begin{aligned}||\phi(X)||_2^2&=\tau^{\prime}(\phi(1)^{1/2}K^*\phi(X^*X)K\phi(1)^{1/2})\\&=\tau^{\prime}(K\phi(1)K^*\phi(X^*X))\\&\leq||\phi(1)||\,||X||_2^2\end{aligned}.\] \end{proof}
	Thus, $\phi$ extends to a bounded linear map from $L^2(\mathcal{M},\tau)$ to $L^2(\mathcal{N},\tau^{\prime})$. This enables taking adjoints. To show that $\phi^*$ is completely positive, we need some technical background.
	\medskip
	
	Recall that for any $n\in\mathbb{N}$, $M_n(\mathcal{M})$ is equipped with the natural trace \[\tau_n\begin{pmatrix}X_{11} & \dots & X_{1n}\\\vdots & &\vdots\\X_{n1} & \dots & X_{nn}\end{pmatrix}=\sum_j \tau(X_{jj}).\] $\phi$ extends naturally to a map $\phi_n:M_n(\mathcal{M})\to M_n(\mathcal{N})$ given by \[\phi_n\begin{pmatrix}X_{11} & \dots & X_{1n}\\\vdots & &\vdots\\X_{n1} & \dots & X_{nn}\end{pmatrix}= \begin{pmatrix}\phi(X_{11}) & \dots & \phi(X_{1n})\\\vdots & &\vdots\\\phi(X_{n1}) & \dots & \phi(X_{nn})\end{pmatrix}.\]
	This is again completely positive and trace preserving, so by Theorem \ref{t1}, bounded. The adjoint $\phi_n^*$ acts on $M_n(\mathcal{N})$ by \[\phi_n^*\begin{pmatrix}Y_{11} & \dots & Y_{1n}\\\vdots & &\vdots\\Y_{n1} & \dots & Y_{nn}\end{pmatrix}= \begin{pmatrix}\phi^*(Y_{11}) & \dots & \phi^*(Y_{1n})\\\vdots & &\vdots\\\phi^*(Y_{n1}) & \dots & \phi^*(Y_{nn})\end{pmatrix}.\] The reader might wonder whether the right hand side makes sense, since $\phi^*(Y_{ij})$ can be unbounded. To work this out, consider a sequence \[\tilde{Z}_m=\begin{pmatrix}Z_{11}^m & \dots & Z_{1n}^m\\\vdots & &\vdots\\Z_{n1}^m & \dots & Z_{nn}^m\end{pmatrix}\] such that $Z_{ij}^m \to \phi^*(Y_{ij})$ in $L^2(\mathcal{M},\tau)$ for all $i,j$. Then $\tilde{Z}_m$ is Cauchy in $L^2(M_n(\mathcal{M}))$ and the matrix \[\begin{pmatrix}\phi^*(Y_{11}) & \dots & \phi^*(Y_{1n})\\\vdots & &\vdots\\\phi^*(Y_{n1}) & \dots & \phi^*(Y_{nn})\end{pmatrix}\] simply denotes its $L^2$ limit. This is unique regardless of the choice of $\tilde{Z}_m$.
	\medskip
	
	With this, it suffices to show that $\phi^*$ is positive, since complete positivity follows by replacing $\mathcal{M}$ and $\mathcal{N}$ with $M_n(\mathcal{M})$ and $M_n(\mathcal{N})$ respectively. Henceforth, the inner product on $L^2(\mathcal{M},\tau)$ will be denoted by $\langle \underline{\hspace{4mm}  },\underline{\hspace{4mm} }\rangle_{\tau}$. This satisfies $\langle X, Y\rangle_{\tau}=\langle Y^*, X^*\rangle_{\tau}$ for all $X, Y\in L^2(\mathcal{M})$, a fact we use in the following theorem.
	\medskip
	
	\begin{thm}\label{t2} Let $\phi:(\mathcal{M},\tau)\to(\mathcal{N},\tau^{\prime})$ be CPTP. Then $\phi^*:L^2(\mathcal{N})\to L^2(\mathcal{M})$ is positive.\end{thm}
	\begin{proof}It suffices to prove that if $A\in\mathcal{N}$ is positive, $\phi^*(A)\in L^2(\mathcal{M})$ is also positive. Let us first show self-adjointness. Note that for any $X\in\mathcal{M}$, \[\begin{aligned}\langle \phi^*(A), X\rangle_{\tau}&=\langle A,\phi(X)\rangle_{\tau^{\prime}}\\&=\tau^{\prime}(A\phi(X)^*)\\&=\langle X^*,\phi^*(A)\rangle_{\tau}\\&=
		\langle \phi^*(A)^*, X\rangle_{\tau}.\end{aligned}\]
	Thus, $\phi^*(A)$ is self-adjoint. Consider the spectral decomposition $$\phi^*(A)=\int_{\mathbb{R}}\lambda\, dP(\lambda).$$ Since $\phi^*(A)$ is affiliated, for each Borel $E\subset\mathbb{R}$, $P(E)\in\mathcal{M}$. Let $P_n=P(-n, -\frac{1}{n})$ for $n\in\mathbb{N}$. Note that $\langle \phi^*(A), P_n\rangle_{\tau}=\langle A, \phi(P_n)\rangle_{\tau}\geq 0$ since $A$ is positive. But \[0\leq\langle \phi^*(A), P_n\rangle_{\tau}=\int_{-n}^{-1/n}\lambda\,d\tau(P(\lambda))\leq -\frac{1}{n}\tau(P_n)\leq 0.\] Since $\tau$ is faithful, $P_n=0$ for all $n$ and therefore, $P(-\infty, 0)=0$, thus completing the proof. \end{proof}
	\begin{cor}Let $B\in\mathcal{M}$ be positive and invertible. Let $\phi:(\mathcal{M},\tau)\to(\mathcal{N},\tau^{\prime})$ be a strict CPTP map. Then the Petz recovery map $\mathcal{R}:L^2(\mathcal{N})\to L^2(\mathcal{M})$ given by $\mathcal{R}(Y)=B^{1/2}\phi^*(\phi(B)^{-1/2}Y\phi(B)^{-1/2})B^{1/2}$ for all $Y\in L^2(\mathcal{N})$ is CPTP.\end{cor}
	\medskip
	
	The DPI for the sandwiched quasi-relative entropies was discussed by Jen\v{c}ov\'a in \cite{jencova}. We give a simple alternate proof in our setting for $\mathcal{S}_2$. This is a lot like Petz's proof of DPI for $f$-divergences \cite{dp}, modified enough to work in infinite dimensions. We start off with a couple of lemmas.
	\medskip
	
	\begin{lem}\label{l1}Let $B\in \mathcal{M}$ be positive and invertible. Let $\phi:(\mathcal{M},\tau)\to(\mathcal{N},\tau^{\prime})$ be strictly CPTP. Consider $V:L^2(\mathcal{N},\tau^{\prime})\to L^2(\mathcal{M},\tau)$ given by $$V(Y)=\phi^*(Y\phi(B)^{-1/2})B^{1/2}$$ for all $Y\in L^2(\mathcal{N})$. Then $V$ is a contraction.\end{lem}
	\begin{proof}Given any $X\in\mathcal{M}$ let $L_X$ and $R_X$ denote the left and right multiplication operators by $X$ respectively. Similarly for $\mathcal{N}$. Write $V=R_{B^{1/2}}\,\phi^*\, R_{\phi(B)^{-1/2}}$. Then $$V^*=R_{\phi(B)^{-1/2}}\,\phi\, R_{B^{1/2}}.$$ Note that $V^*$ maps $\mathcal{M}$ to $\mathcal{N}$.  Let $X\in\mathcal{M}$ and note that \[||V^*(X)||_2^2=\tau^{\prime}(\phi(XB^{1/2})\phi(B)^{-1}\phi(B^{1/2}X^*).\label{e7}\tag{7}\]
	The block matrix \[\begin{pmatrix}XX^* & XB^{1/2}\\B^{1/2}X^* & B\end{pmatrix}\] is positive and therefore,  \[\begin{pmatrix}\phi(XX^*) & \phi(XB^{1/2})\\\phi(B^{1/2}X^*) & \phi(B)\end{pmatrix}\,\geq\, O.\] This implies $$\phi(XB^{1/2})\phi(B)^{-1}\phi(B^{1/2}X^*)\leq\phi(XX^*).$$ Taking trace and using $\eqref{e7}$, $$||V^*(X)||_2\leq||X||_2.$$ Thus $V^*$, and hence, $V$, is a contraction. \end{proof}
	\medskip
	
	\begin{lem}\label{l2} Let $\phi:(\mathcal{M},\tau)\to(\mathcal{N},\tau^{\prime})$ be strictly CPTP. Let $\Delta:L^2(\mathcal{M})\to L^2(\mathcal{M})$ and $\Delta_0:L^2(\mathcal{N})\to L^2(\mathcal{N})$ be given by $\Delta=L_BR_{B^{-1}}$ and $\Delta_0=L_{\phi(B)}R_{\phi(B)^{-1}}$ respectively and let $V$ be as in lemma \ref{l1}. Then $V^*\Delta V\leq \Delta_0$. \end{lem}
	\begin{proof}Note that $$V^*\Delta V=R_{\phi(B)^{-1/2}}\,\phi\, L_{B}\,\phi^*\, R_{\phi(B)^{-1/2}}.$$ Hence, $V^*\Delta V\leq \Delta_0$ is equivalent to $\phi\,L_b\,\phi^*\leq L_{\phi(B)}$. This happens if and only if $\phi^*L_{\phi(B)^{-1}}\phi\leq L_{B^{-1}}$. To prove this we note that for any $X\in\mathcal{M}$, \[\begin{aligned}&\begin{pmatrix}X^*B^{-1}X & X^*\\X & B\end{pmatrix}\geq O\\&\implies\begin{pmatrix}\phi(X^*B^{-1}X) & \phi(X)^*\\\phi(X) & \phi(B)\end{pmatrix}\geq O.\end{aligned}\] This implies $$\tau^{\prime}(\phi(X)^*\phi(B)^{-1}\phi(X))\leq\tau(X^*B^{-1}X),$$ which means $$\langle \phi^*\,L_{\phi(B)^{-1}}\phi\, X, X\rangle_{\tau}\leq \langle L_{B^{-1}}X, X\rangle_{\tau}$$ for all $X\in\mathcal{M}$, completing the proof.\end{proof}
	\medskip
	
	Given $p\geq 1$ and a positive, invertible $B\in(\mathcal{M},\tau)$ consider the Araki-Masuda $p$-norm on $\mathcal{M}$ given by $$||X||_{B, p}=||B^{-1/2q}XB^{-1/2q}||_p$$ for all $X\in\mathcal{M}$, where $q$ is the harmonic conjugate of $p$. This is equivalent to the $L^p$ norm on $\mathcal{M}$, and extends to an equivalent norm on $L^p(\mathcal{M})$. Let $\tau(B)=1$. Then for any positive $A\in L^p(\mathcal{M})$ with $\tau(A)=1$, $$\mathcal{S}_p(A|B)=||A||_{B,p}^p.$$ For $p=2$, $||.||_{B,2}$ gives a Hilbert space norm on $L^2(\mathcal{M})$ induced by the inner product $$\langle X, Y\rangle_B=\langle X, B^{-1/2}YB^{-1/2}\rangle_{\tau}.$$  
	\medskip
	
	Note that if $\phi:\mathcal{M}\to\mathcal{N}$ is strictly CPTP, its adjoint when seen as a map from $(L^2(\mathcal{M}),\, ||.||_{B,2})$ to $(L^2(\mathcal{N}),\,||.||_{\phi(B), 2})$ is precisely the Petz recovery map (see $\eqref{e4}$) induced by $B$. As we see later, this is a crucial observation.
	\medskip
	
	For now, let us complete the proof of DPI for $\mathcal{S}_2$.
	\medskip
	
	\begin{thm}\label{t3}Let $\phi:(\mathcal{M},\tau)\to(\mathcal{N},\tau^{\prime})$ be strictly CPTP and let $B\in\mathcal{M}$ such that $B$ is positive, invertible and $\tau(B)=1$. Then for all $X\in\mathcal{M}$ $||\phi(X)||_{\phi(B),\,2}\leq||X||_{B,\,2}$.\end{thm}
	\begin{proof}This is equivalent to showing that \[\langle \phi^*L_{\phi(B)^{-1/2}}R_{\phi(B)^{-1/2}}\phi(X), X\rangle_{\tau}\leq\langle L_{B^{-1/2}}R_{B^{-1/2}} X, X\rangle_{\tau}\] for all $X\in\mathcal{M}$, which is the same as saying the block matrix \[\begin{pmatrix}L_{B^{-1/2}}R_{B^{-1/2}} & \phi^*\\ \phi & L_{\phi(B)^{1/2}}R_{\phi(B)^{1/2}}\end{pmatrix}\] is positive as an operator on $L^2(\mathcal{M},\tau)\oplus L^2(\mathcal{N},\tau^{\prime})$. But this is equivalent to \[\phi\, L_{B^{1/2}}R_{B^{1/2}}\,\phi^*\leq L_{\phi(B)^{1/2}}R_{\phi(B)^{1/2}}.\label{e8}\tag{8}\]  Recall the operators $V$, $\Delta$ and $\Delta_0$ from lemmas \ref{l1} and \ref{l2}. Note that by operator concavity of the square root, \[\Delta_0^{1/2}\geq V^*\Delta^{1/2}V\] Plugging in the expressions of $V$, $\Delta$ and $\Delta_0$ in terms of left and right multiplication operators, we get $\eqref{e8}$.
	\end{proof}
	\medskip
	
	Theorem \ref{t3} states any strict CPTP map between tracial von-Neumann algebras is a contraction with respect to the Araki-Masuda $L^2$-norms. A simple, important fact about contractions on a Hilbert space will be instrumental in the proof of our main result, and we state it as the next lemma.
	\medskip
	
	\begin{lem}\label{l3} Let $\mathcal{H}$ and $\mathcal{K}$ be Hilbert spaces and let $T:\mathcal{H}\to\mathcal{K}$ be a contraction. Then for any $x\in\mathcal{H}$, $$||x-T^*Tx||^2\leq ||x||^2-||Tx||^2.$$\end{lem}
	\begin{proof}Note that \[\begin{aligned}||x-T^*Tx||^2&=||x||^2+\langle (T^*T)^2x, x\rangle-2||Tx||^2\\&\leq ||x||^2-||Tx||^2\end{aligned}\] since $(T^*T)^2\leq T^*T$.\end{proof}
	As an immediate application, we get :
	\begin{thm}\label{t4}Let $\phi:(\mathcal{M},\tau)\to(\mathcal{N},\tau^{\prime})$ be a strict CPTP map and let $B\in\mathcal{M}$ be positive and invertible with $\tau(B)=1$. Then for any positive $A\in L^2(\mathcal{M})$ with $\tau(A)=1$, \[||A-\mathcal{R}(\phi(A))||_{B,\,2}^2\leq\mathcal{S}_2(A|B)-\mathcal{S}_2(\phi(A)|\phi(B))\] where $\mathcal{R}$ is the Petz recovery map $\eqref{e4}$.\end{thm}
	\begin{proof}Note that $\mathcal{R}$ is the adjoint of $\phi$ with respect to the Araki-Masuda $2$-norms induced by $B$ and $\phi(B)$ and apply lemma \ref{l3}. \end{proof}
	\medskip
	
	Theorem \ref{t4} already gives us recoverability : $\mathcal{S}_2(A|B)=\mathcal{S}_2(\phi(A)|\phi(B))$ if and only if $\mathcal{R}(\phi(A))=A$. To get inequality $\eqref{e5}$, we need one more result.
	\medskip
	
	\begin{thm}\label{t5} For any $X\in L^2(\mathcal{M})$ and positive, invertible $B$ with $\tau(B)=1$, we have $||X||_1^2\leq||X||_{B,2}^2$.\end{thm}
	\begin{proof}Enough to prove for $X\in\mathcal{M}$. Note that $||X||_{B,2}^2\ \langle L_{B^{-1/2}}R_{B^{-1/2}}\, X, X\rangle_{\tau}.$ By operator A.M-G.M inequality and anti-monotonicity of the inverse, $$L_{B^{-1/2}}R_{B^{-1/2}}\geq 2(L_B+R_B)^{-1}.$$ So \[||X||_{B,\, 2}^2\geq 2\langle (L_B+R_B)^{-1}X, X\rangle_{\tau}.\label{e9}\tag{9}\] Then for any unitary $U\in\mathcal{M}$, \[\begin{aligned}|\langle X, U\rangle_{\tau}|^2 &=|\langle (L_B+R_B)^{-1/2}X, (L_B+R_B)^{1/2}U\rangle_{\tau}|^2\\&\leq 2\langle (L_B+R_B)^{-1}X, X\rangle_{\tau}\,\frac{\langle BU+UB, U\rangle_{\tau}}{2}\end{aligned}\] by Cauchy-Schwarz. By $\eqref{e9}$ and cyclicity of trace, $|\langle X, U\rangle_{\tau}|^2\leq ||X||_{B,\,2}^2\tau(B)=||X||_{B,\,2}^2.$ Taking supremum over unitaries $U$, $$||X||_1=\sup \{|\langle X,U\rangle|_{\tau}:U\,\mathrm{unitary}\}\leq||X||_{B,2}.$$ \end{proof}
	\medskip
	
	\begin{cor}\label{c2}Let $\phi:(\mathcal{M},\tau)\to(\mathcal{N},\tau^{\prime})$ be strictly CPTP and let $B\in\mathcal{M}$ be positive, invertible with $\tau(B)=1$. Then for any positive $A\in L^2(\mathcal{M})$ with $\tau(A)=1$, $$||A-\mathcal{R}\circ\phi(A)||^2_1\leq [\mathcal{S}_2(A|B)-\mathcal{S}_2(\phi(A)|\phi(B))].$$\end{cor}
	\begin{proof}Follows directly from Theorems \ref{t4} and \ref{t5}.\end{proof}
	\medskip
	
	\begin{rem2}Let $\psi_{A}$ and $\psi_{\mathcal{R}(\phi(A))}$ denote the normal states on $\mathcal{M}$ corresponding to $A$ and $\mathcal{R}(\phi(A))$ respectively. Then $||\psi(A)-\psi_{\mathcal{R}(\phi(A))}||=||A-\mathcal{R}\circ\phi(A)||_1$. Applying corollary \ref{c2}, $||\psi(A)\,-\,\psi_{\mathcal{R}(\phi(A))}||^2\leq [\mathcal{S}_2(A|B)-\mathcal{S}_2(\phi(A)|\phi(B))]$. This shows that the states $\psi_A$ and $\psi_{\mathcal{R}(\phi(A))}$ are close in the functional norm whenever the relative entropy difference $[\mathcal{S}_2(A|B)-\mathcal{S}_2(\phi(A)|\phi(B))]$ is small. \end{rem2}
	\medskip
	
	The transition probability $P(\psi_1|\psi_2)$ between two states $\psi_1$ and $\psi_2$ on a $C^*$-algebra $\mathcal{A}$ is defined by $\eqref{e6}$. This was introduced by Uhlmann in \cite{uhl2}, and studied further in \cite{alb, uhal}. The {\it fidelity}  between $\psi_1$ and $\psi_2$ is defined as $$F(\psi_1|\psi_2)=\sqrt{P(\psi_1|\psi_2)}.$$ Let $\psi_1$ and $\psi_2$ be normal states on a tracial von-Neumann algebra $(\mathcal{M},\tau)$ given by $\psi_1(X)=\tau(AX)$ and $\psi_2(X)=\tau(BX)$ for all $X\in\mathcal{M}$, where $A, B\in L^1(\mathcal{M})$ are positive with $\tau(A)=\tau(B)=1$. It follows from Uhlmann's work \cite{uhl2} that \[F(\psi_1|\psi_2)=\sup_{U\in\mathcal{M},\, U\, \mathrm{unitary}} |\langle UA^{1/2}, B^{1/2}\rangle_{\tau}|.\label{e10}\tag{10}\]
	\medskip
	
	The following recoverability bound is an immediate consequence :
	\medskip
	
	\begin{thm}\label{t6}Let $\phi:(\mathcal{M},\tau)\to(\mathcal{N},\tau^{\prime})$ be strictly CPTP. Let $B\in \mathcal{A}$ be positive and invertible. Then for any positive $A\in L^2(\mathcal{M})$ with $\tau(A)=1$, \[4[1-F(\psi_{A}|\psi_{\,\mathcal{R}(\phi(A)}))]^2\leq||A-\mathcal{R}\circ\phi(A)||^2_1\leq [\mathcal{S}_2(A|B)-\mathcal{S}_2(\phi(A)|\phi(B))]\] where $\psi_A$ and $\psi_{\,\mathcal{R}(\phi(A)})$ are the normal states corresponding to $A$ and $\mathcal{R}(\phi(A))$ respectively.\end{thm}
	\begin{proof}The right hand side is just corollary \ref{c2}. To prove the left hand side, note that for any two $X, Y\in L^1(\mathcal{M})$ such that $X$ and $Y$ are positive with $\tau(X)=\tau(Y)=1$, $$||X^{1/2}-Y^{1/2}||_2^2\leq||X-Y||_1$$ by the Powers-St\o rmer inequality (see \cite{aude, pst}). But by $\eqref{e10}$, \[\begin{aligned}2[1-F(\psi_X|\psi_Y)]&=\inf_{U\in\mathcal{M},\, U\, \mathrm{unitary}}||X^{1/2}-UY^{1/2}||_2^2\\&\leq ||X^{1/2}-Y^{1/2}||^2\\&\leq ||X-Y||_1\end{aligned}\] where $\psi_X$ and $\psi_Y$ are the normal states corresponding to $X$ and $Y$ respectively. Squaring both sides, we are done.\end{proof}
	\medskip
	
	This gives the required generalization of inequality $\eqref{e3}$ to infinite dimensions.
	\medskip

	\section{Appendix}
	
	Here we present a brief discussion on the Uhlmann fidelity for two states on a $C^*$-algebra $\mathcal{A}$. Recall $\eqref{e6}$, where the {\it transition probability} for two states $\psi_1,\,\psi_2:\mathcal{A}\to\mathbb{C}$ is defined as \[
	\begin{aligned}P(\psi_1|\psi_2)
		&:= \sup_{\pi\in\mathrm{Hom}(\mathcal{A}, B(\mathcal{H}))}\{|\langle x, y\rangle|^2:\psi_1(A)=\langle\pi(A)x,x\rangle,\,\\&\psi_2(A)=\langle\pi(A)y, y\rangle,\, A\in\mathcal{A},\,||x||=||y||=1 \}.
	\end{aligned}\]
	
	The fidelity is given by $F(\psi_1|\psi_2)=\sqrt{P(\psi_1|\psi_2)}$. For density matrices $A$ and $B$, this reduces to the usual fidelity $\mathrm{tr}|A^{1/2}B^{1/2}|$. 
	\medskip
	
	It is easy to see that $F$ is symmetric in its arguments and $$0\leq F(\psi_1|\psi_2)\leq 1.$$ Alberti obtained a concrete expression for $F(\psi_1|\psi_2)$ in \cite{alb}, where he showed that if $\pi:\mathcal{A}\to B(\mathcal{H})$ is a representation such that $\psi_1(A)=\langle\pi(A)x, x\rangle$ and $\psi_2(A)=\langle\pi(A)y, y\rangle$ for unit vectors $x,y\in\mathcal{H}$, \[F(\psi_1|\psi_2)=\sup_{U\in\pi(\mathcal{A})^{\prime},\, U\, \mathrm{unitary}} |\langle x, Uy\rangle|.\label{e11}\tag{11}\] Here, $\pi(\mathcal{A})^{\prime}$ is the commutant of $\pi(\mathcal{A})$ in $B(\mathcal{H})$. The following is an immediate consequence.
	\begin{thm}\label{t7}Let $\mathcal{A}$ be a $C^*$-algebra and let $S(\mathcal{A})$ denote the state space of $\mathcal{A}$. Then the map $d:S(\mathcal{A})\times S(\mathcal{A})\to[0,\infty]$ given by $d(\psi_1,\psi_2)=\cos^{-1} F(\psi_1|\psi_2)$ is a metric on $S(\mathcal{A})$.\end{thm}
	\begin{proof}Symmetry is obvious. Also, $d(\psi_1,\psi_2)=0$ if and only if $F(\psi_1|\psi_2)=1$, so from Uhlmann \cite{uhl2}, it follows that $\psi_1=\psi_2$.
		\medskip
		
	For the triangle inequality, consider three states $\psi_1,\psi_2,\psi_3$. Choose a representation $\pi:\mathcal{A}\to B(\mathcal{H})$ and unit vectors $x, y, v\in\mathcal{H}$ such that $\psi_1(A)=\langle\pi(A)x, x\rangle$, $\psi_2(A)=\langle \pi(A)y, y\rangle$, $\psi_3(A)=\langle \pi(A)v, v\rangle$ for all $A\in\mathcal{A}$.
	\medskip
	
	By $\eqref{e11}$, $$d(\psi_1,\psi_2)=\inf_{U\in\pi(\mathcal{A})^{\prime},\, U\, \mathrm{unitary}} \cos^{-1} |\langle x, Uy\rangle|.$$ Obtain similar expressions for $d(\psi_1,\psi_3)$ and $d(\psi_2,\psi_3)$ respectively. Let $U_1, U_2$ be unitaries in $\pi(\mathcal{A})^{\prime}$ and let $W=U_1U_2^*$. Then $W\in\pi(\mathcal{A})^{\prime}$.
	\medskip
	
	Then \[\begin{aligned}&\cos^{-1}|\langle x, U_1v\rangle|+\cos^{-1}|\langle y, U_2v\rangle|\\&=\cos^{-1}|\langle x, U_1v\rangle|+\cos^{-1}|\langle Wy,\, U_1v\rangle|.\end{aligned} \] Note that for unit vectors $v_1$ and $v_2$, $\cos^{-1}|\langle v_1,v_2\rangle|$ denotes the Fubini-Study distance between the rank one projections $v_1\otimes v_1^*$ and $v_2\otimes v_2^*$, which is a metric. Thus, \[\begin{aligned}&\cos^{-1}|\langle x, U_1v\rangle|+\cos^{-1}|\langle y, U_2v\rangle|\\&=\cos^{-1}|\langle x, U_1v\rangle|+\cos^{-1}|\langle Wy,\, U_1v\rangle|\\&\geq\cos^{-1}|\langle x, Wy\rangle|\\&\geq\cos^{-1} F(\psi_1|\psi_2).\end{aligned}\] Taking infimum over unitaries $U_1,\,U_2\in\pi(\mathcal{A})^{\prime}$, we are done.\end{proof}
	\medskip
	
	\begin{rem2}The metric in Theorem \ref{t7} is the $C^*$-algebra analogue of the {\it Bures angle} between density matrices, arising as the geodesic distance with respect to the SLD Fisher information (see \cite{jenc}).\end{rem2}
	\medskip
	
	Uhlmann in \cite{uhl2} showed that $P(\psi_1|\psi_2)$ is concave under {\it Gibbsian mixtures}, which means for three states $\psi_1,\psi_2,\psi:\mathcal{A}\to\mathbb{C}$ and any $\lambda\in[0,1]$, $$P(\lambda\psi_1+(1-\lambda)\psi_2|\psi)\geq \lambda P(\psi_1|\psi)+(1-\lambda)P(\psi_2|\psi).$$ This demonstrates concavity in each argument. 
	\medskip
	
	 Alberti also proved that for any unital completely positive $\phi:\mathcal{A}\to\mathcal{B}$, and states $\psi_1,\psi_2:\mathcal{B}\to\mathbb{C}$, \[F(\psi_1\circ\phi|\psi_2\circ\phi)\geq F(\psi_1|\psi_2).\label{e12}\tag{12}\]
	\medskip
	
	It is to be noted that $\eqref{e12}$ is equivalent to joint concavity in the matrix case (see \cite{uhl}). However, it is not obvious in the general $C^*$-algebraic setup since we cannot substitute the action of $\phi$ with a CPTP map on density matrices. Nevertheless, Farenick and Rahaman \cite{far} proved joint concavity in the setting of tracial $C^*$-algebras, using a variational formula for the fidelity similar to the finite dimensional case.
	\medskip
	
	Here, we offer an elementary proof of joint concavity of the fidelity in an arbitrary $C^*$-algebra. To the best of our knowledge, this has not appeared explicitly before in literature, though the result itself is common folklore. The technique is adapted from Uhlmann's original proof of separate concavity of $P(\psi_1|\psi_2)$ in \cite{uhl2}.
	\medskip
	
	\begin{thm} The fidelity $F(\psi_1|\psi_2)$ is jointly concave in the state space of a $C^*$-algebra $\mathcal{A}$.\end{thm}
	\begin{proof}Let $\psi_1,\psi_2$ and $\rho_1, \rho_2$ be pairs of states on $\mathcal{A}$ and let $\lambda\in [0,1]$. We need to show that \[F(\lambda\psi_1+(1-\lambda)\psi_2|\lambda\rho_1+(1-\lambda)\rho_2)\geq \lambda F(\psi_1|\rho_1)+(1-\lambda)F(\psi_2|\rho_2).\] Choose representations $\pi_1:\mathcal{A}\to B(\mathcal{H})$ and $\pi_2:\mathcal{A}\to B(\mathcal{K})$, and unit vectors $x_1, y_1\in\mathcal{H}$, $x_2, y_2\in\mathcal{K}$ such that $$\psi_1(A)=\langle \pi_1(A)x_1, x_1\rangle,\,\rho_1(A)=\langle \pi_1(A)y_1, y_1\rangle$$ and $$\psi_2(A)=\langle \pi_2(A)x_2, x_2\rangle,\,\rho_2(A)=\langle \pi_2(A)y_2, y_2\rangle$$ for all $A\in\mathcal{A}$.
		Consider the direct sum $\tilde{\pi}=\pi_1\oplus\pi_2$ and vectors $$\tilde{x}=\begin{pmatrix}\sqrt{\lambda}x_1\\[3pt] \sqrt{1-\lambda}x_2\end{pmatrix}$$, $$\tilde{y}=\begin{pmatrix}\sqrt{\lambda}y_1\\[3pt] \sqrt{1-\lambda}y_2\end{pmatrix}$$ in $\mathcal{H}\oplus\mathcal{K}$. Note that $$|\langle \tilde{x},\tilde{y}\rangle|=|\lambda\langle x_1, y_1\rangle+(1-\lambda)\langle x_2,y_2\rangle|.$$ Now $$\lambda \psi_1(A)+(1-\lambda)\psi_2(A)=\langle (\pi_1\oplus\pi_2)(A)\tilde{x}, \tilde{x}\rangle$$ and $$\lambda \rho_1(A)+(1-\lambda)\rho_2(A)=\langle (\pi_1\oplus\pi_2)(A)\tilde{y}, \tilde{y}\rangle.$$ Hence by $\eqref{e6}$, \[|\lambda\langle x_1, y_1\rangle+(1-\lambda)\langle x_2,y_2\rangle|\leq F(\lambda\psi_1+(1-\lambda)\psi_2|\lambda\rho_1+(1-\lambda)\rho_2).\label{e13}\tag{13}\] Note that the expressions for $\lambda\psi_1+(1-\lambda)\psi_2$ and $\lambda\rho_1+(1-\lambda)\rho_2$ $\eqref{e13}$ do not change if we replace $x_1$ with $e^{i\theta_1}x_1$ and $x_2$ with $e^{i\theta_2}x_2$ for any $\theta_1,\,\theta_2\in\mathbb{R}$. Hence, $\eqref{e13}$ becomes \[|\lambda\,e^{i\theta_1}\langle x_1, y_1\rangle+(1-\lambda)\,e^{i\theta_2}\langle x_2,y_2\rangle|\leq F(\lambda\psi_1+(1-\lambda)\psi_2|\lambda\rho_1+(1-\lambda)\rho_2).\] Maximizing over $\theta_1,\,\theta_2$, $$\lambda|\langle x_1, y_1\rangle|+(1-\lambda)|\langle x_2,y_2\rangle|\leq F(\lambda\psi_1+(1-\lambda)\psi_2|\lambda\rho_1+(1-\lambda)\rho_2).$$ This holds for any pair of GNS representations $\pi_1$ and $\pi_2$, so taking supremum over all such choices and using $\eqref{e6}$, \[F(\lambda\psi_1+(1-\lambda)\psi_2|\lambda\rho_1+(1-\lambda)\rho_2)\geq \lambda F(\psi_1|\rho_1)+(1-\lambda)F(\psi_2|\rho_2).\]
	\end{proof}
	\begin{acknowledgement}
		I thank my PhD supervisor Prof. Tanvi Jain for her comments and suggestions to improve the paper. 
	\end{acknowledgement}
	\medskip
	
	\textbf{Conflict of interest :} The author declares no conflict of interest.
	\medskip
	
	\textbf{Data availability statement :} No available data has been used.
	\medskip
	
	\textbf{Funding :} I thank Indian Statistical Institute for the financial support during my PhD via their institute fellowship.
		\bibliographystyle{amsplain}
		
	\end{document}